\pdfoutput=1
\documentclass[UKenglish,a4paper,pdfa,cleveref,thm-restate]{lipics-v2021}

\usepackage{csquotes,fixcmex,mathtools,pgfplots,xfrac}
\pgfplotsset{compat=1.18}
\usepgfplotslibrary{colormaps}

\title{Map-Matching Queries under Fréchet Distance on Low-Density Spanners}
\author{Kevin Buchin}{Department of Computer Science, TU Dortmund,
Germany\and\url{https://ls11-www.cs.tu-dortmund.de/staff/buchin}}{%
kevin.buchin@tu-dortmund.de}{https://orcid.org/0000-0002-3022-7877}{}
\author{Maike Buchin}{Faculty of Computer Science, Ruhr-Universität Bochum,
Germany\and\url{https://informatik.rub.de/buchin/}}{maike.buchin@rub.de}{%
https://orcid.org/0000-0002-3446-4343}{}
\author{Joachim Gudmundsson}{School of Computer Science, University of Sydney,
Australia\and\url{https://www.sydney.edu.au/engineering/about/our-people/academic-staff/joachim-gudmundsson.html}}{%
joachim.gudmundsson@sydney.edu.au}{https://orcid.org/0000-0002-6778-7990}{}
\author{Aleksandr Popov}{Department of Mathematics and Computer Science,
TU Eindhoven, The Netherlands\and\url{https://apopov.eu/}}{%
alex@apopov.eu}{https://orcid.org/0000-0002-0158-1746}{Supported by the Dutch
Research Council (NWO) under the project number 612.001.801.}
\author{Sampson Wong}{Department of Computer Science, University of Copenhagen,
Denmark\and\url{https://sites.google.com/view/sampsonwong/}}{%
sawo@di.ku.dk}{https://orcid.org/0000-0003-3803-3804}{Supported in part by
Starting Grant 1054-00032B from the Independent Research Fund Denmark under the
Sapere Aude research career programme.}
\authorrunning{K. Buchin, M. Buchin, J. Gudmundsson, A. Popov, and S. Wong}
\hideLIPIcs
\nolinenumbers
\newcommand*{\fr}{d_\mathrm{F}}
\newcommand*{\R}{\mathbb{R}}
\newcommand*{\bO}{\mathcal{O}}
\newcommand*{\eps}{\varepsilon}
\newcommand*{\opt}{\mathrm{opt}}
\newcommand*{\trough}{\mathrm{trough}}
\newcommand*{\dsh}{---}
\DeclareMathOperator*{\argmin}{arg\,min}
\DeclareMathOperator*{\argmax}{arg\,max}

\newtheorem{problem}[theorem]{Problem}
\graphicspath{{nfigs/}}

\begin{document}
\maketitle

\begin{abstract}
Map matching is a common task when analysing GPS tracks, such as vehicle
trajectories.
The goal is to match a recorded noisy polygonal curve to a path on the map,
usually represented as a geometric graph.
The Fréchet distance is a commonly used metric for curves, making it a natural
fit.
The map-matching problem is well-studied, yet until recently no-one tackled the
data structure question: preprocess a given graph so that one can query the
minimum Fréchet distance between all graph paths and a polygonal curve.
Recently, Gudmundsson, Seybold, and Wong~\cite{prev} studied this problem for
arbitrary query polygonal curves and \(c\)-packed graphs.
In this paper, we instead require the graphs to be \(\lambda\)-low-density
\(t\)-spanners, which is significantly more representative of real-world
networks.
We also show how to report a path that minimises the distance efficiently
rather than only returning the minimal distance, which was stated as an open
problem in their paper.
\end{abstract}

\section{Introduction}\label{sec:intro}
Location data is ubiquitous, and analysis of that data is a common task.
GPS trajectories of vehicles or people often suffer from being noisy or
skipping large portions of the movement.
To analyse them more precisely, one may use \emph{map matching.}
The idea is that the vehicles move on a road network, so one could snap their
trajectories to a road network in a way that most closely resembles the
original.
We assume that the trajectory is a polygonal curve and the map is a geometric
graph in the plane.

The map-matching problem has received considerable attention, with multiple
surveys comparing the various approaches on different types of
data~\cite{refs_acm,refs_survey1,refs_survey2,refs_survey4,refs_survey5,refs_survey6}.
The \emph{Fréchet distance} is a natural measure of similarity for polygonal
curves~\cite{alt_godau,godau}, taking into account the ordering of the points
of the polygonal curves and capturing the maximal distance along them.
There is a host of work considering map matching specifically under the Fréchet
distance~\cite{frechet_graph,refs_frechet1,refs_frechet2,refs_frechet3,refs_frechet5,refs_frechet4,refs_frechet6},
including a seminal paper by Alt, Efrat, Rote, and Wenk~\cite{frechet_graph}.
Their algorithm requires \(\bO(mn\log mn \log n)\) time and \(\bO(mn)\) space
to match a polygonal curve of length~\(m\) to a planar graph \(G = (V, E)\)
with complexity \(\lvert V\rvert + \lvert E\rvert = n\).
As shown via a conditional lower bound by Gudmundsson, Seybold, and
Wong~\cite{prev}, this query time is close to optimal for planar graphs: there
is no algorithm that runs in \(\bO((mn)^{1 - \delta})\) time for any
\(\delta > 0\) that solves this problem after polynomial-time preprocessing of
the graph.
However, real-world road networks are rarely planar due to the presence of
bridges and tunnels, so we would like to find other assumptions on the graph
that also allow for faster query times.

Chen, Driemel, Guibas, Nguyen, and Wenk~\cite{refs_frechet3} study the
map-matching problem under \emph{realistic input} assumptions, which aim to
exclude particular types of degenerate instances to provide stronger results.
In their work in particular, the graph has low density and the trajectory is
\(c\)-packed.
A polygonal curve is called \emph{\(c\)-packed} if in any ball of radius~\(r\),
the total length of the curve inside the ball is at most~\(cr\).
We can use a similar definition for geometric graphs by measuring the total
length of the edges inside the ball.
Unfortunately, \(c\)-packedness is a strong assumption that is difficult to
satisfy.
We instead define low-density graphs, which are more representative of
real-world networks.

\begin{definition}\label{def:density}
A geometric graph \(P = (V, E)\) is \emph{\(\lambda\)-low
density}~\cite{ld_search,ld_thesis} if for every disk of radius~\(r > 0\) in
the plane, there are at most \(\lambda\) edges of length at least~\(2r\)
intersecting the disk.
\end{definition}

Our work has no assumptions on the trajectories at the expense of stricter
assumptions on the maps.
Most often one will have a large number of trajectories being mapped to a
relatively complex network, so to avoid the steep dependency on network
complexity when matching every trajectory, we consider the query version of the
problem.
We preprocess the map so that we can quickly answer many map-matching queries,
where each query is a trajectory.
To our knowledge, this problem has only been studied on \(c\)-packed
graphs~\cite{prev,refs_frechet4}.
Gudmundsson and Smid~\cite{refs_frechet4} show an approach for \(c\)-packed
trees with long edges and query trajectories with long edges.
Gudmundsson et al.~\cite{prev} assume that the graph is \(c\)-packed, but do
not impose any restrictions on the query trajectories.
However, \(c\)-packedness is not a realistic assumption for graphs representing
road networks.
Consider the example map of \cref{fig:map}: it is not \(c\)-packed for any
constant~\(c\), as that would require the total length of the roads be at
most~\(cr\) in all disks of radius~\(r\), and it is instead often much closer
to~\(cr^2\).
On some scale, this problem arises with many road networks, including city
streets or motorways.
Therefore, we would like to devise an approach with more realistic assumptions
on the graph.

\begin{figure}[tb]
\centering
\includegraphics{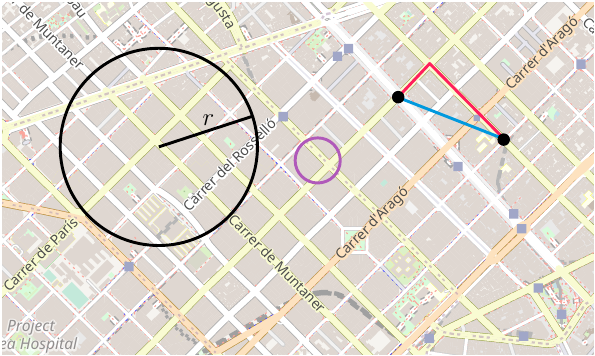}
\caption{An example road network, in Barcelona.
The total road length in a disk of radius~\(r\) is closer to \(cr^2\)
than~\(cr\), so this road network is not \(c\)-packed.
However, the number of long edges intersecting the purple disk is small, and
the red path is not much longer than the blue path, so the network is
\(\lambda\)-low density and a \(t\)-spanner for small \(\lambda\)~and~\(t\).
Map data from OpenStreetMap~\cite{openstreetmap}.}
\label{fig:map}
\end{figure}

We instead assume that our graph has low density, defined above.
As observed by Chen et al.~\cite{refs_frechet3}, the value of density does not
grow with the considered area on the map, and road networks typically have low
density; it is also a strictly weaker assumption than packedness.
We further assume that the graph is a \emph{\(t\)-spanner.}
A geometric graph is a \(t\)-spanner if for any two vertices, the length of the
shortest path between them in the graph is at most \(t\) times larger than the
Euclidean distance between them.
Road networks, in particular in urban areas, are typically good
spanners~\cite{feedlinks2011,Nordbeck64}.
For the example of \cref{fig:map}, it is clear that the road network is
a \(t\)-spanner and \(\lambda\)-low density for low \(t\)~and~\(\lambda\).
Compared to previous work, our assumptions make the approach significantly more
applicable on real-life road networks.

In this paper, we solve the \emph{map-matching query problem} under realistic
assumptions:
\begin{problem}\label{prob:mmq}
Given a geometric graph \(P\), construct a data structure that can answer the
following queries: for a polygonal curve~\(Q\) in~\(\R^2\),
\begin{enumerate}
\item compute \(\min_\pi \fr(\pi, Q)\) and
\item report \(\argmin_\pi \fr(\pi, Q)\),
\end{enumerate}
where \(\pi\) ranges over all paths between two vertices in~\(P\) and \(\fr\)
denotes the Fréchet distance.
\end{problem}
\pagebreak

We present a \((1 + \eps)\)-approximation under the assumptions listed above.
Our approach differentiates from previous work by Gudmundsson et
al.~\cite{prev} in two key aspects:
\begin{itemize}
\item we require the graph~\(P\) to be \(\lambda\)-low density and a
\(t\)-spanner, rather than \(c\)-packed, which is a more realistic assumption
for a road network~\cite{feedlinks2011,refs_frechet3,Nordbeck64};
\item we solve the problem of reporting the path that minimises the Fréchet
distance, which was stated as an open problem in their paper.
\end{itemize}

In order to achieve these results, we have to use different techniques, albeit
at the cost of a \(\sqrt{n}\) factor replacing a polylogarithmic factor in the
running time and space.
Where the paper by Gudmundsson et al.~\cite{prev} uses a semi-separated pair
decomposition, we construct a hierarchy of small balanced separators and store
appropriate associated data to guide the search for the optimal Fréchet
distance.
A \emph{balanced separator} of a graph \(P = (V, E)\) is a set of vertices
\(S \subseteq V\) that splits \(P\) into connected components of size at most
\(c\cdot\lvert V\rvert\) for constant~\(c\).
Combining the changes in analysis and the capability to report a path, we get
the following result.
Here \(n = \lvert V\rvert + \lvert E\rvert\) for a geometric graph
\(P = (V, E)\), and \(m\) is the number of vertices on the query polygonal
curve.

\begin{restatable}{theorem}{thmmain}\label{thm:main}
Suppose we are given a \(\lambda\)-low-density \(t\)-spanner of
complexity~\(n\) and a fixed \(0 < \eps < 1\).
Let \(\chi = \sfrac{1}{\eps^2}\log\sfrac{1}{\eps}\) and let
\(\varphi = (\sfrac{\lambda}{\eps^3} + \sfrac{t^2}{\eps^2})^2\).
In expected time \(\bO(\lambda\chi^2 n^{\sfrac{5}{2}} \log n)\) and using
\(\bO(\lambda \chi^2 n^{\sfrac{3}{2}})\) space, we can construct a data
structure \((1 + \eps)\)-approximating \cref{prob:mmq} that
performs distance queries in time
\(\bO(\varphi\cdot\sfrac{\lambda}{\eps} \cdot m\sqrt{n}\log mn \cdot
(\log^2 n + \varphi\log n + \varphi\cdot\sfrac{\lambda}{\eps}))\), and
answers the reporting queries for a path of length~\(\ell\) in
\(\bO(\sfrac{\ell}{\eps})\) additional time.
\end{restatable}

In a typical setting, \(\lambda\)~and~\(t\) are small constants.
We need
\(\bO(n^{\sfrac{3}{2}} \cdot \sfrac{1}{\eps^4} \cdot \log^2\sfrac{1}{\eps})\)
space, resolving to \(\bO(n^{\sfrac{3}{2}})\) for fixed~\(\eps\), and we
support distance queries in time \(\bO(\eps^{-7} m\sqrt{n}\log mn \cdot
(\log^2 n + \eps^{-6} \log n + \eps^{-7}))\), or
\(\bO(m\sqrt{n}\log mn \cdot \log^2 n)\) for fixed~\(\eps\).

The rest of the paper is organised as follows.
After covering some preliminaries in \cref{sec:prelims}, we first tackle the
simpler problem of finding a path in the graph that most closely follows a line
segment between two vertices of the graph in terms of the Fréchet distance in
\cref{sec:sp}.
In that setting, we find a \(3\)-approximation.
The data structure we develop there is used later for obtaining a search window
for when the end points of the path are not given.
In \cref{sec:segm}, we generalise this to an arbitrary query line segment that
does not have to start or end at a graph vertex, and show how to achieve a
\((1 + \eps)\)-approximation.
We also describe how to report a path that corresponds to a
\((1 + \eps)\)-approximation.
Finally, in \cref{sec:mmq}, we show how to combine the segment queries in order
to handle a complete polygonal curve.

\section{Preliminaries}\label{sec:prelims}
In this paper, we work with geometric graphs, that is, graphs embedded in the
plane with straight-line edges.
For a graph~\(P = (V, E)\), we denote its complexity
with~\(n = \lvert V\rvert + \lvert E\rvert\).

We denote the fact that a path \(\pi\) goes from \(u \in V\) to \(v \in V\)
by \(\pi: u \leadsto v\); and \(\pi \circ \sigma\) denotes the composition
(concatenation) of paths~\(\pi\),~\(\sigma\).
Denote the Euclidean distance between points \(x, y \in \R^2\) with
\(\lVert x - y\rVert\).
We can consider the edges weighted by defining the weight of an edge
\(e \in E\) as
\(\lvert e\rvert = \lVert u - v\rVert\) for \(e\) connecting \(u, v \in V\).
We define the \emph{graph distance \(d_P\)} as the shortest path distance along
the graph between any two vertices \(u, v \in V\):
\(d_P(u, v) = \sum_{e \in \pi} \lvert e\rvert\), where \(\pi: u \leadsto v\) is
a shortest path in \(P\) between \(u\) and \(v\).

We assume our input graph has low density in this paper, as defined previously.
We need to define two more graph properties that we use in our construction.
\begin{definition}\label{def:lanky}
A graph \(P = (V, E)\) is \emph{\(\tau\)-lanky}~\cite{lanky} if for every disk
of radius \(r > 0\) centred at any vertex \(v \in V\), there are at most
\(\tau\) edges of length at least \(r\) that are cut by the disk.
\end{definition}
Here an edge is \emph{cut} by a disk if exactly one of its endpoints is inside
the disk.
It is easy to see that a \(\tau\)-lanky graph has bounded degree of at most
\(\tau\); and that any \(\lambda\)-low-density graph is also \(\lambda\)-lanky.
Let us also formally define a \(t\)-spanner.
\begin{definition}\label{def:spanner}
A graph \(P = (V, E)\) is called a \emph{\(t\)-spanner} if for any two vertices
\(u, v \in V\), we have \(d_P(u, v) \leq t \cdot \lVert u - v\rVert\).
\end{definition}

A query is a polygonal curve in the plane, that is, a sequence of points in
\(\R^2\) connected with line segments.
For a query curve \(Q\), let \(m\) be the number of points in the sequence.

\section{Straight Path Queries}\label{sec:sp}
In this \lcnamecref{sec:sp}, we present a \(3\)-approximation to the following
problem, so that for the value \(r\) that we return, we have
\(\min_\pi \fr(\pi, uv) \leq r \leq 3 \cdot \min_\pi \fr(\pi, uv)\).

\begin{problem}\label{prob:sp}
Given a geometric graph \(P = (V, E)\), construct a data structure that can
answer the following queries: for a pair of vertices \(u, v \in V\), compute
\(\min_\pi \fr(\pi, uv)\), where \(\pi: u \leadsto v\) is a path in \(P\).
\end{problem}

Let \(n = \lvert V\rvert + \lvert E\rvert\).
To solve \cref{prob:sp} efficiently, we impose an additional constraint
on~\(P\)\dsh we require in this \lcnamecref{sec:sp} that \(P\) has a graph
property satisfying two criteria:
\begin{enumerate}
\item the property is decreasing monotone, so it holds on all induced
subgraphs;
\item and any graph with the property admits a small separator.
\end{enumerate}

An example of such a property is planarity: any subgraph of a planar graph is
planar, and the existence of small separators in planar graphs is a classical
result~\cite{planar_sep}.
However, not all road networks are planar, as most road networks include
bridges and tunnels.
Instead, we require that \(P\) is \(\tau\)-lanky~\cite{lanky}.
It is trivial to show that any subgraph of a \(\tau\)-lanky graph is also
\(\tau\)-lanky; and Le and Than~\cite{lanky} show that a \(\tau\)-lanky graph
of complexity \(n\) admits a balanced separator of size \(\bO(\tau\sqrt{n})\)
that can be found in \(\bO(\tau n)\) expected time.

In \cref{sec:segm}, we show how to generalise the query to arbitrary segment
endpoints and how to improve the result to a \((1 + \eps)\)-approximation for
any fixed \(\eps > 0\).

\subparagraph*{Intuition.}
When constructing the data structure, we can use the algorithm by Alt et
al.~\cite{frechet_graph} in order to compute the Fréchet distance between a
line segment and a path in the graph.
At query time, running that algorithm would be prohibitively slow.
We also want to achieve subquadratic storage, so we cannot precompute the
distances for all pairs of vertices.

Broadly, the idea is to find sufficient structure in the graph to be able to
find a small set of vertices so that any path in the graph passes through at
least one of these vertices; we call them \emph{transit} vertices.
Then we can precompute the distances between the optimal path and the line
segment when going from any vertex of the graph to one of the transit vertices.
At query time, we then only need to find an optimal transit vertex.
Since we are composing two paths, the computed distance is only a
\(3\)-approximation.

More specifically, a balanced separator in a graph forms a set of transit
vertices.
We can compute them hierarchically and store the separators and the precomputed
distances in a binary tree.
With some organisation, at query time, we can efficiently find all the relevant
transit vertices\dsh the ones that may separate the two query vertices on a
path.

\subparagraph*{Data structure.}
We construct a hierarchy of separators on the graph and store it with some
extra information in a binary tree.
Each node in the tree represents both an induced subgraph of \(P\), and a
separator of that induced subgraph.
Consider the node \(i\) corresponding to some induced subgraph
\(P_i = (V_i, E_i)\) of \(P\).
Conceptually, the node represents the balanced separator \(S_i\), so the subset
of vertices of \(P_i\), splitting it into two subgraphs \(A_i\) and \(B_i\). 

The root stores \(S_1\) and the extra information for all pairs from
\(V \times S_1\), so the top-level balanced separator for the entire graph.
The two children of each node correspond to the subgraphs \(A_i\) and \(B_i\).
The recursion ends when the subgraphs in the leaves have constant size.
In a leaf \(i\), assign \(S_i = V_i\), so compute the distances for all pairs
of vertices.

For every pair of vertices \((u, s) \in V_i \times S_i\), we store
\(\min_\pi \fr(\pi, us)\), where \(\pi: u \leadsto s\) in \(P\).
(Note that a path \(\pi\) may leave \(P_i\).)
We call all vertices in \(S_i\) \emph{transit} vertices; and all pairs
\((u, s)\) for which we store the distances are called \emph{transit pairs.}

In addition, for each vertex \(u \in V\), we store a pointer to the tree node
\(i\) so that \(u \in S_i\).
There is exactly one such node in the tree for every vertex: if a vertex is
part of a separator, it will not be in an induced subgraph further down in the
recursion, and if it is never chosen to be in a separator, then it belongs to a
leaf, which is treated as \(S_i\) in its entirety.

\subparagraph*{Construction.}
We construct the hierarchy top--down, computing the separators on the induced
subgraphs at every level using the result of Le and Than~\cite{lanky}.
For each transit pair \((p, s)\) in a node, we compute the appropriate Fréchet
distance in the entire graph using the algorithm by Alt et
al.~\cite{frechet_graph}, extended by Gudmundsson et
al.~\cite[Lemma~4.3]{prev}.
As we construct the separators, we also store in a table the pointer for each
vertex to the correct node.

\subparagraph*{Distance query.}
Suppose the query is to find the minimal Fréchet distance between the segment
\(uv\) and any path between \(u\) and \(v\).
Initialise \(\opt = \infty\).
First, we use the table to find the pointers to the two nodes in the tree \(i\)
and \(j\) so that \(u \in S_i\) and \(v \in S_j\).
Then we find their lowest common ancestor, call it node \(a\).
For every node \(a'\) on the path from \(a\) to the root of the tree, perform
the following procedure, updating~\(\opt\).
At the end, return~\(\opt\).

Denote \(D_{xy} = \min_\pi \fr(\pi, xy)\) over all \(\pi: x \leadsto y\).
For the query \(uv\), denote \(D'_x = \min_{r \in uv} \lVert r - x\rVert\), so
the shortest distance between \(x\) and any point on \(uv\).
For all \(s \in S_{a'}\), fetch the stored \(D_{us}\) and \(D_{sv}\) and
compute \(D'_s\).
Then compute \(D = \max(D_{us}, D_{sv}) + D'_s\) and finally assign
\(\opt = \min(\opt, D)\).

\subparagraph*{Running time analysis.}
For the distance queries, we take \(\bO(1)\) time to follow the pointers;
\(\bO(\log n)\) time to find the lowest common ancestor; and then \(\bO(1)\)
time to check the distance per transit vertex.
As we check all transit vertices on the path from the lowest common ancestor to
the root, we can write down the worst-case recurrence as
\[T(k) = T(\sfrac{2k}{3}) + \bO(\tau\sqrt{k})\]
for a graph on \(k\) vertices, since the balanced separator we use subdivides
the graph into two subgraphs on at most \(\sfrac{2k}{3}\) vertices.
For the entire graph, this resolves to \(\bO(\tau\sqrt{n})\).
This term dominates the query time.

For the construction, we need \(\bO(\tau k)\) expected time to find a separator
in a graph of size~\(k\).
In each node, for each of \(\bO(\tau k\sqrt{k})\) transit pairs, we compute the
distance in \(\bO(n \log n)\) time.
Assuming we build the tree until the leaves are of constant size, we get the
recurrence
\begin{align*}
T(k) &= T(\sfrac{2k}{3}) + T(\sfrac{k}{3}) +
\bO(\tau k + \tau k\sqrt{k}\cdot n\log n)\\
&= T(\sfrac{2k}{3}) + T(\sfrac{k}{3}) + \bO(\tau k\sqrt{k}\cdot n\log n)\,,
\end{align*}
which resolves to \(\bO(\tau n^{\sfrac{5}{2}} \log n)\) expected time overall.

\subparagraph*{Space.}
We store a table of pointers of size \(\bO(n)\) and the main data structure.
For a graph on \(k\) vertices, we store constant-size data for each transit
pair; and there are \(\bO(\tau k\sqrt{k})\) transit pairs.
Overall, the space used resolves to \(\bO(\tau n\sqrt{n})\), as follows from
the recurrence
\[T(k) = T(\sfrac{2k}{3}) + T(\sfrac{k}{3}) + \bO(\tau k\sqrt{k})\,.\]

\subparagraph*{Correctness.}
It remains to show that the described query procedure gives us an appropriate
distance.
First, assume that we do consider an optimal transit vertex; we show that we
indeed compute a \(3\)-approximation.
The following proof is essentially given by Gudmundsson et
al.~\cite[Theorem~4.1]{prev}, and relies on a semi-separated pair
decomposition~\cite[Lemma~5.5]{approx} rather than separators.
We include the proof for the sake of completeness and ease of reading.

\begin{lemma}
Suppose that \(\opt = \min_{\pi'} \fr(\pi', uv)\) for query \(uv\) and
\(\pi': u \leadsto v\), and that \(\pi = \argmin_{\pi'} \fr(\pi', uv)\) passes
through a transit vertex \(s\), so \(\pi: u \leadsto s \leadsto v\).
Let \(s'\) be the transit vertex that minimises
\(D = \max(D_{us'}, D_{s'v}) + D'_{s'}\).
If \(s\) is considered when finding \(D\), then
\(\opt \leq D \leq 3\cdot\opt\).
\end{lemma}
\begin{proof}
Let \(t\) be the point on \(uv\) closest to \(s'\).
Let \(\pi_{us'} = \argmin_{\pi'} \fr(\pi', us')\) over \(\pi': u \leadsto s'\),
and define \(\pi_{s'v}\) similarly.
Note that the composition of these paths \(\pi_{us'} \circ \pi_{s'v}\) does not
have to be the same as \(\pi\).
Then
\begin{align*}
\opt = \fr(\pi, uv) &\leq \fr(\pi_{us'} \circ \pi_{s'v}, uv)
\leq\max\big(\fr(\pi_{us'}, ut), \fr(\pi_{s'v}, tv)\big)\\
&\leq\max\big(\fr(\pi_{us'}, us') + \lVert s' - t\rVert,
\fr(\pi_{s'v}, s'v) + \lVert s' - t\rVert\big)\\
&= \lVert s' - t\rVert + \max\big(\fr(\pi_{us'}, us'),
\fr(\pi_{s'v}, s'v)\big)\\
&= D'_{s'} + \max(D_{us'}, D_{s'v}) = D\,.
\end{align*}

On the other hand, note that \(D \leq \max(D_{us}, D_{sv}) + D'_s\), as \(s'\)
minimises that expression.
We also have \(D'_s \leq \fr(\pi, uv)\).
Let \(\pi(u, s)\) be the subpath of \(\pi\) from \(u\) to \(s\).
Let \(r\) be the point in \(uv\) that is aligned to \(s\) in the Fréchet
alignment between \(\pi\) and \(uv\).
Then
\begin{align*}
D_{us} = \fr(\pi_{us}, us) \leq \fr(\pi(u, s), us) &\leq \fr(\pi(u, s), ur) +
\fr(ur, us)\\
&= \fr(\pi(u, s), ur) + \lVert r - s\rVert\\
&\leq \fr(\pi, uv) + \fr(\pi, uv) = 2\fr(\pi, uv)\,.
\end{align*}
Using the same argument for \(D_{sv}\), we conclude
\[D \leq D'_s + \max(D_{us}, D_{sv}) \leq \fr(\pi, uv) + 2\fr(\pi, uv) =
3\fr(\pi, uv)\,,\]
and so \(\opt \leq D \leq 3\cdot\opt\) and the value is a \(3\)-approximation.
\end{proof}

\begin{figure}
\centering
\includegraphics{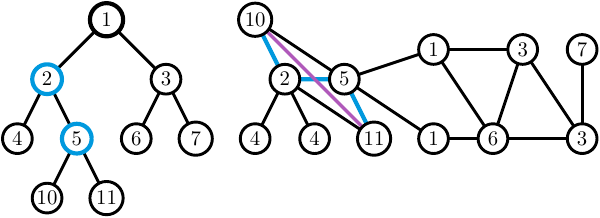}
\caption{A representation of the hierarchy (left) for the graph (right).
A query segment is shown in purple, a possible path in blue.
We check nodes \(5\), \(2\), and \(1\).
If we pick the transit vertex in \(5\), then the path may be
\(10 \to 2 \to 5\), so we may need to go up the tree to find the next transit
pair.}
\label{fig:graph}
\end{figure}

Now we show that we consider all the relevant transit vertices.
See \cref{fig:graph}.
\begin{lemma}
For the query \(uv\), the procedure considers a transit vertex \(s\) such that
\(s\) lies on an optimal path \(\pi = \argmin_{\pi'} \fr(\pi', uv)\), where
\(\pi': u \leadsto v\).
\end{lemma}
\begin{proof}
We consider two cases based on where the lowest common ancestor is found.
First, suppose that the lowest common ancestor \(a\) contains \(u\), \(v\), or
both \(u\) and \(v\) in the separator.
In other words, \(u \in S_a\) or \(v \in S_a\), and so \(s = u\) or \(s = v\).
Then \(\pi\) passes through \(s\), and we consider \(s\) as a transit vertex.

Now assume that the lowest common ancestor \(a\) does not contain \(u\) or
\(v\) in the separator; then \(u\) and \(v\) are separated by \(S_a\).
Without loss of generality, let \(u \in A_a\) and \(v \in B_a\).
If the path \(\pi\) stays within the subgraph \(P_a\), then it goes
through some \(s \in S_a\), which we consider.
Otherwise, it goes through some separator between \(P_a\) and the rest of the
graph; we check exactly all the vertices in these separators, as they fall on
the path from \(a\) to the root.
\end{proof}

Bringing the above considerations together, we get the main result of this
\lcnamecref{sec:sp}.
\begin{theorem}\label{thm:sp}
Given a \(\tau\)-lanky graph of complexity \(n\), we can construct a data
structure giving a \(3\)-approximation for \cref{prob:sp} in expected time
\(\bO(\tau n^{\sfrac{5}{2}} \log n)\),
using \(\bO(\tau n\sqrt{n})\) space, that supports distance queries in time
\(\bO(\tau\sqrt{n})\).
\end{theorem}

\section{Map-Matching Segment Queries}\label{sec:segm}
In this \lcnamecref{sec:segm}, we generalise the construction we just presented
to compute a \((1 + \eps)\)-approximation and to handle reporting, as well as
to support arbitrary query line segments.

\begin{problem}\label{prob:segm}
Given a geometric graph \(P = (V, E)\), construct a data structure that can
answer the following queries: for a line segment \(pq\) in the plane,
\begin{enumerate}
\item compute \(\min_\pi \fr(\pi, pq)\) and
\item report \(\argmin_\pi \fr(\pi, pq)\),
\end{enumerate}
where \(\pi\) ranges over all paths between two vertices in \(P\).
\end{problem}

For distance queries, we closely follow the work of Gudmundsson et
al.~\cite{prev}, which in turn follows the approach of Driemel and
Har-Peled~\cite{approx}.
The latter appears at first to be devoted to a rather different problem, but
turns out to be very helpful in the map-matching setting.

\subparagraph*{Data structure for distance queries with fixed path endpoints.}
We can immediately use the approach of \cref{sec:sp} on arbitrary segments.
Suppose the query is a pair of vertices \(u, v\in V\) and a segment
\(pq \subset \R^2\).
Then we can find a \(3\)-approximation to \(\min_\pi \fr(\pi, pq)\), where
\(\pi: u\leadsto v\).
To achieve that, we just need to define \(D'_{s'} = \fr(us' \circ s'v, pq)\)
and let \(t\) be the point on \(pq\) aligned with \(s'\) under this matching.

We can directly use the following statement~\cite[Lemma~5.2]{prev}, which
closely mimics the approach of Driemel and Har-Peled~\cite[Lemma~5.8]{approx}:
\begin{lemma}\label{lem:grid}
Let \(u, v \in V\) be a fixed pair of vertices.
Let \(\eps > 0\) and \(\chi = \sfrac{1}{\eps^2}\log\sfrac{1}{\eps}\).
In \(\bO(\chi^2 n\log n)\) time and using \(\bO(\chi^2)\) space, one can
construct a data structure that, given a query segment \(pq\) in the plane,
returns in \(\bO(1)\) time a \((1 + \eps)\)-approximation to
\(\min_\pi \fr(\pi, pq)\), where \(\pi: u \leadsto v\).
\end{lemma}
The idea behind this \lcnamecref{lem:grid} is to construct an exponential grid
around both fixed vertices so that the grid is denser close to the vertices.
There is an upper bound and a lower bound on how far the grid goes, which is
based on \(\eps\) and \(\min_\pi \fr(\pi, uv)\).
If the segment \(pq\) is closer to \(uv\) than the smallest grid cell, then
taking \(\min_\pi \fr(\pi, uv)\) gives us a good approximation; if the segment
is very far, then \(\fr(pq, uv)\) dominates.
Otherwise, we are guaranteed that there are grid points \(p'\) and \(q'\) that
match \(p\) and \(q\) closely with respect to \(u\) and \(v\).
We can simply precompute \(\min_\pi \fr(\pi, p'q')\) for all pairs of points
\(p'\) and \(q'\) and return an appropriate value in constant time when given a
query.
See \cref{fig:report}.

In order to improve the approximation ratio to \(1 + \eps\), we use Lemma~5.3
by Gudmundsson et al.~\cite{prev}, substituting our data structure of
\cref{thm:sp} for their data structure of Lemma~5.1.
The argument is the same: we can construct a grid around each graph vertex and
precompute the distances for all pairs of grid vertices for each transit pair;
and we can store that in the data structure of \cref{thm:sp}.
At query time, when testing each transit vertex \(s\), we subsample the
relevant part of segment
\(pq\) with \(\bO(\sfrac{1}{\eps})\) points to find an optimal spot that should
align with \(s\).
We use the data structure of \cref{thm:sp} to make sure we do not need to
sample too many points.
With our time and space bounds, we get the following
\lcnamecref{lem:grid_all}.
\begin{lemma}\label{lem:grid_all}
Let \(\eps > 0\) and \(\chi = \sfrac{1}{\eps^2}\log\sfrac{1}{\eps}\).
In expected time \(\bO(\tau \chi^2 n^{\sfrac{5}{2}} \log n)\) and using
\(\bO(\tau \chi^2 n\sqrt{n})\) space, we can construct a data structure that,
given a query segment \(pq\) in the plane and a pair of vertices
\(u, v \in V\), returns in \(\bO(\sfrac{\tau}{\eps}\sqrt{n})\) time a
\((1 + \eps)\)-approximation to \(\min_\pi \fr(\pi, pq)\), where
\(\pi: u \leadsto v\).
\end{lemma}

\subparagraph*{Reporting a path.}
Next we discuss the modifications needed to report a curve that realises the
\((1 + \eps)\)-approximate distance.
We can perform an approximate distance query first.
Once we have the distance, we can find the transit pairs that realise it; with
these pairs, we can store the next vertex on the optimal path.
We can then repeat these queries with the new pairs.
It remains to show how to perform this sequence of queries consistently,
i.e.\@ so that an approximate route for a subpath also approximates the
complete path.

Recall that when computing the \((1 + \eps)\)-approximation, we consider a ball
of a certain radius around a transit vertex, and we take
\(\bO(\sfrac{1}{\eps})\) sample points on the query segment \(pq\) inside the
ball, to test the Fréchet alignment with the transit vertex.
The first modification is that we impose fixed coordinates for the sample
points.
For some fixed point on the line~\(pq\) and for some constant~\(c\), every
sample point is \(\bO(\sfrac{c}{\eps})\) distance away from the fixed point.

Next, we describe the necessary modifications to the data structure of
\cref{lem:grid_all}.
With each transit pair and for each pair of grid points, in addition to the
Fréchet distance, we also store the first vertex on the optimal path, so
\(u' \in V\) such that for \(\pi' = \argmin_\pi \fr(\pi, pr)\), we have
\(\pi': u \to u' \leadsto s\).
(Here \(r\) is the point on \(pq\) that maps to \(s\).)

The query proceeds as follows.
First, we perform the distance query for \(pq\) and record the optimal transit
vertex \(s\).
Find the point \(r\) among the \(\bO(\sfrac{1}{\eps})\) samples on \(pq\) that
aligns with \(s\).
Query the pairs \((u, s)\) and \((s, v)\) with \(pr\) and \(rq\), respectively,
and retrieve the stored adjacent vertices \(u'\) and \(v'\).
Again, find the optimal alignment points on \(pr\) and \(rq\); find pairs
\((u'', s)\) and \((s, v'')\); repeat until the complete path is reported.
In the special case when \(s = v\) or \(s = u\), only one sequence of queries
has to be performed.
If \(u\) and \(v\) are both in a leaf, we can proceed as if \(s = v\).
See \cref{fig:report}.

\begin{figure}
\centering
\includegraphics{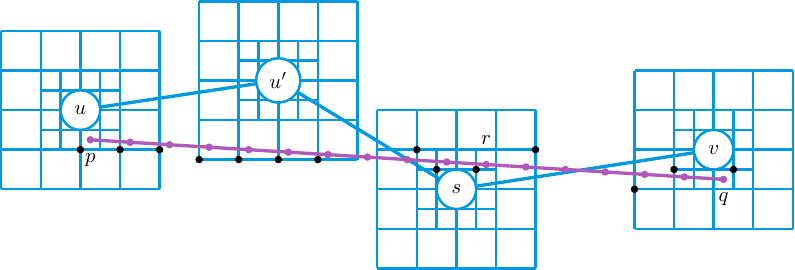}
\caption{A query trajectory \(pq\) is shown in purple, and the reported path in
the graph is shown in blue.
We sample points on \(pq\) at regular distance and snap them to the exponential
grid around the graph vertices.
Once we find \(r\) on \(pq\) that aligns with the transit vertex \(s\), we can
query the pair \((u, s)\) with the (snapped) segment \(pr\) to find the next
vertex \(u'\).}
\label{fig:report}
\end{figure}

Suppose the transit vertex \(s\) is stored in some \(S_i\).
If the optimal path leaves \(P_i\), then it is possible that \(u'\) is not in
\(P_i\), and so the pair \((u', s)\) is not stored in node \(i\).
However, then \(u'\) must be in some separator separating \(P_i\) from a
different subgraph \(P_j\).
Furthermore, note that the separator in question must be on the path from \(i\)
to the root.
Thus, we can go up until we find \(u' \in S_j\) for some \(j < i\).
We can continue the procedure, now for the transit pair \((s, u')\), finding
some \(s'\) so that the path is of the shape \(s \to s' \leadsto u'\).
See \cref{fig:graph,fig:report}.

We briefly analyse the time and space bounds.
We find the next vertex on the path from the free-space diagram when computing
the Fréchet distance.
As we store constant extra information, the preprocessing and space bounds are
unchanged.
For the query time, in addition to the distance query, we report a path of
length \(\ell\).
To find each next vertex, we find the correct transit pair in constant time,
then test \(\bO(\sfrac{1}{\eps})\) alignment options.
We may have to go up the tree; however, as we never go down the tree, that
traversal happens only once per query.
Therefore, the extra time needed to report the path with \(\ell\) vertices is
\(\bO(\log n + \sfrac{\ell}{\eps})\).
It remains to show that the reported path indeed corresponds to a
\((1 + \eps)\)-approximation.

\begin{lemma}
For query \(pq\), if \(\pi = \argmin_{\pi^*} \fr(\pi^*, pq)\) has the shape
\(\pi: u \to u' \leadsto s\), and \(u'\) is aligned to some \(p' \in pq\) under
the Fréchet alignment, then \(\pi' = \argmin_{\pi^*} \fr(\pi^*, p'q)\) of the
shape \(\pi': u' \leadsto s\) is a subpath of \(\pi\).
\end{lemma}
\begin{proof}
Without loss of generality, we can assume that \(s\) is a transit vertex.
If the distances were computed exactly, the statement would clearly hold.
We need to show that the sampling and the grid do not introduce
inconsistencies.

Recall that the sample points are placed on the line segment independently of
context.
Therefore, the location of sample points is the same on \(pq\) and \(p'q\).
Furthermore, we always snap these original sample points to the grid, and the
grid does not depend on the path.
Therefore, we can view \(pq\) as a sequence of grid points that all possible
sample points would snap to; and \(p'q\) then snaps to a subsequence of those
grid points.
For the pairs of grid points, the distances are computed directly.
Therefore, we do not introduce any additional error, compared to a distance
query, and so the reported path corresponds to a \((1 + \eps)\)-approximation.
\end{proof}

\subparagraph*{Data structure for finding the path endpoints.}
So far, we only required that \(P\) is \(\tau\)-lanky.
For the next data structure, we also need \(P\) to be a \(t\)-spanner.
Recall that \(d_P(u, v)\) denotes the shortest path distance in the graph
between the vertices \(u\) and \(v\).
If \(P = (V, E)\) is a \(t\)-spanner, then for any \(u, v \in V\), we have
\(d_P(u, v) \leq t \cdot \lVert u - v\rVert\).
To solve \cref{prob:segm}, we need one more data structure.
\Cref{lem:grid_all} still requires us to pick vertices \(u\) and \(v\) to check
the paths \(\pi: u \leadsto v\).
We need to be able to select a subset of candidate vertices so that we can
obtain a \((1 + \eps)\)-approximation, but the subset is still small enough.
To that aim, we perform the same procedure as Gudmundsson et al.~\cite{prev},
but get different bounds.

In particular, we run Gonzalez's \(k\)-centre clustering
algorithm~\cite{gonzalez} for \(k = n\) on the vertices of the graph
\(P = (V, E)\) using the distance \(d_P\).
In short, the algorithm selects cluster centres from \(V\) iteratively,
starting with a random one; and each following one is the furthest away from
any other centre.
Let \(c_1\) be the first (random) centre.
Define the \emph{radius} of a clustering to be the maximum distance from any
vertex to its closest centre.
Denote \(C_i = \{c_1, \dots, c_i\}\).
Then for all \(2 \leq i \leq n\), we compute
\[c_i = \argmax_{v \in V} \min_{c \in C_{i - 1}} d_P(v, c)\,,\qquad
r_i = \max_{v \in V} \min_{c \in C_i} d_P(v, c)\,.\]
We obtain a sequence \(\langle (C_1, r_1), \dots, (C_n, r_n)\rangle\), where
\(r_n = 0\), since all vertices are centres.
Using this sequence, we can show the following \lcnamecref{lem:clustered}.

\begin{lemma}\label{lem:clustered}
Let \(P = (V, E)\) be a \(t\)-spanner and let \(S\) be a square in the plane
with side length \(2r\).
Then there exists a set of vertices \(T \subseteq V\) satisfying two
properties:
\begin{enumerate}
\item \(\lvert T\rvert = \bO((\sfrac{t}{\eps})^2)\), and
\item for all \(v \in V \cap S\), there is \(z \in T\) such that
\(d_P(v, z) \leq \eps r\).
\end{enumerate}
\end{lemma}
\begin{proof}
If \(r_1 < \eps r\), pick \(T = \{c_1\}\); then the first property holds
immediately, and the second property holds by definition of~\(r_1\).
Otherwise, note that \(r_n = 0\), so now let \(i\) be the index so that
\(r_i \geq \eps r\) and \(r_{i + 1} < \eps r\).
Take \(S'\) to be the square concentric with \(S\), but with the side length of
\(4r\), and let \(T = C_{i + 1} \cap S'\).
The second property is immediately satisfied, since any vertex, including those
in \(S\), is closer than \(\eps r\) to some centre; and choosing \(S'\) this
way ensures that we cannot exclude any relevant centres, since \(\eps < 1\).

To see that the first property is true, consider \(C_i\).
Due to the sequence of picking centres, we know that any two vertices \(c_j\)
and \(c_\ell\) with \(j < \ell\) in \(T\) are far apart, i.e.\@
\(d_P(c_j, c_\ell) \geq \eps r\).
If this were not true, then \(c_\ell\) would not have been chosen as a centre,
since there still are vertices in \(V\) that are further than \(\eps r\) away
from any centre.
But then we know
\[\eps r \leq d_P(c_j, c_\ell) \leq t \cdot \lVert c_j - c_\ell\rVert\,,\]
so any two points are at least \(\sfrac{\eps}{t} \cdot r\) apart in the plane.
In a square with side length \(4r\), we can only pack
\(\bO((\sfrac{t}{\eps})^2)\) of these vertices.
\(C_{i + 1}\) has only one more vertex, so the first property holds for \(T\),
as well.
\end{proof}

Gudmundsson et al.~\cite{prev} show how to construct a data structure based on
their version of \cref{lem:clustered}.
The proof is the same; only the bounds change.

\enlargethispage*{2\baselineskip}
\begin{lemma}\label{lem:cluster_ds}
Let \(P = (V, E)\) be a \(t\)-spanner, and let \(0 \leq \eps \leq 1\).
In \(\bO(n^2 \log n)\) time and using \(\bO(n \log n)\) space, we can construct
a data structure that, given a query square \(S\) in the plane with side length
\(2r\), returns a set of vertices \(T\) satisfying \cref{lem:clustered} in time
\(\bO(\log n + (\sfrac{t}{\eps})^2))\).
\end{lemma}
\pagebreak

The main \lcnamecref{thm:segm} of this \lcnamecref{sec:segm} follows using our
data structures.

\begin{theorem}\label{thm:segm}
Given a \(\tau\)-lanky \(t\)-spanner of complexity \(n\), we can construct the
data structure for \cref{prob:segm} in expected time
\(\bO(\tau \eps^{-4} \log^2(\sfrac{1}{\eps}) n^{\sfrac{5}{2}} \log n)\) and
using \(\bO(\tau \eps^{-4} \log^2(\sfrac{1}{\eps}) n\sqrt{n})\) space, so the
distance queries can be answered in time
\(\bO(\tau t^8 \eps^{-9} \sqrt{n} \log n (\log n + \sfrac{\tau}{\eps}))\); and
the reporting queries for a path of length \(\ell\) can be answered in
\(\bO(\sfrac{\ell}{\eps})\) additional time.
\end{theorem}
\begin{proof}
The changes we made for reporting do not affect \cref{lem:cluster_ds}, so the
proof of Gudmundsson et al.~\cite{prev} applies. 
We only discuss the time bounds.
Preprocessing simply consists of building the data structures for
\cref{lem:grid_all,lem:cluster_ds}.
For the query time, consider first the decision version of the algorithm.
We query the data structure of \cref{lem:cluster_ds} twice; and then for every
pair of possible matching vertices, we query the data structure of
\cref{lem:grid_all}.
The second step dominates, taking \(\bO(\tau t^4 \eps^{-5} \sqrt{n})\) time.

For the optimisation version, we use parametric search with \(N_P = \sqrt{n}\)
parallel processors.
The sequential version runs in the same time as the decision version, so
\(T_S = \bO(\tau t^4 \eps^{-5} \sqrt{n})\).
In the parallel version, querying the distance data structure can be done with
\(\sqrt{n}\) processors, each performing \(\bO(\sfrac{\tau}{\eps})\) amount of
work, then combining the values to find the minimum in \(\bO(\log n)\) time.
Thus, \(T_P = \bO((\sfrac{t}{\eps})^4 \cdot (\sfrac{\tau}{\eps} + \log n))\).
The time for the optimisation version is now
\(\bO(N_P T_P + T_P T_S \log N_P)\).
This amounts to
\(\bO(\tau t^8 \eps^{-9} \sqrt{n} \log n (\log n + \sfrac{\tau}{\eps}))\).

For the reporting query, perform the distance query and record the optimal path
endpoints; then, as we discussed, it costs extra \(\bO(\sfrac{\ell}{\eps})\)
time to report a path of length \(\ell\).
\end{proof}

\section{General Map-Matching Queries}\label{sec:mmq}
In this \lcnamecref{sec:mmq}, we generalise the problem again to handle a
polygonal curve rather than a line segment as a query.
The procedure is very similar; however, we want to make sure that the Fréchet
alignment between a query curve and a path can align vertices of the query to
points on graph edges, and not just to graph vertices.
To that effect, we need to extend \cref{lem:clustered} so we can sample a small
number of points on graph edges.

Gudmundsson et al.~\cite{prev} use \(c\)-packedness again; however, in our
setting, the \(t\)-spanner property is not sufficient, as it does not give us
guarantees about the graph distance between points on the edges.
Here we require the graph to also be \(\lambda\)-low density.

\begin{lemma}\label{lem:ld_points}
Let \(P = (V, E)\) be a \(\lambda\)-low-density \(t\)-spanner, let
\(F = \{f \in \R^2 \mid f \in e, e \in E\}\), and let \(S\) be a square in the
plane with side length \(2r\).
Then there exists a set of points \(T \subseteq F\) satisfying two
properties:
\begin{enumerate}
\item \(\lvert T\rvert = \bO(\sfrac{t^2}{\eps^2} + \sfrac{\lambda}{\eps^3})\),
and
\item for all \(p \in F \cap S\), there is \(z \in T\) such that
\(d_P(p, z) \leq \eps r\).
\end{enumerate}
\end{lemma}
\begin{proof}
We can use \cref{lem:clustered} with \(\eps' = \sfrac{\eps}{2}\) to obtain the
set \(T_1\) of size \(\bO((\sfrac{t}{\eps'})^2)\) so that for all
\(v \in V \cap S\), \(d_P(p, z) \leq \eps' r\) for some \(z \in T_1\).
Define \(E_r \subseteq E\) to contain the edges of length at least \(\eps r\).
Let \(S'\) be a square concentric with \(S\) but with the side length \(4r\).
For each \(e \in E_r\), choose \(\bO(\sfrac{1}{\eps})\) evenly spaced points
on \(e \cap S'\) with the distance between them of at most \(\eps r\).
Let \(T_2\) be the set of all such points, and assign \(T = T_1 \cup T_2\).

We first show that the first property holds.
For \(T_2\), we need to bound the size of \(E_r \cap S'\).
As \(P\) is \(\lambda\)-low density, we know that there are at most \(\lambda\)
edges of length at least \(\eps r\) intersecting any disk of diameter
\(\eps r\), and every edge has \(\bO(\sfrac{1}{\eps})\) sample points.
We can cover \(S'\) with \(\bO((\sfrac{1}{\eps})^2)\) such disks, so
\(\lvert T_2\rvert = \bO(\lambda (\sfrac{1}{\eps})^3)\).
Therefore,
\(\lvert T\rvert = \bO(\sfrac{t^2}{\eps^2} + \sfrac{\lambda}{\eps^3})\).

Now consider the second property.
Note that \(V \subset F\).
For any \(v \in V \cap S\), we immediately conclude that the property holds by
\cref{lem:clustered}.
For any \(p \in e \cap S\), \(e \in E\) with \(\lvert e\rvert \leq \eps r\),
note that both endpoints of \(e\) lie in \(S'\).
So there is a vertex \(v \in V \cap S'\) so that \(d_P(p, v) \leq \eps' r\),
and by \cref{lem:clustered}, \(d_P(v, z) \leq \eps' r\) for some \(z \in T_1\).
Therefore, \(d_P(p, z) \leq 2\eps' r = \eps r\).
Finally, for any \(p \in e \cap S\), \(e \in E_r\), it is clear that there is a
point not further than \(\eps r\) in \(T_2\).
\end{proof}

We will build a data structure similar to Gudmundsson et al.'s~\cite{prev}.
We start by stating a definition of \(\lambda\)-low density in \(\R^3\).

\begin{definition}\label{def:ld_3d}
A set of objects in \(\R^3\) is \(k\)-low density if, for every axis-parallel
cube \(H_r\) with side length \(r\), there are at most \(k\) objects of size at
least \(r\) that intersect \(H_r\).
The size of an object is the side length of its smallest axis-parallel
enclosing cube.
\end{definition}

\begin{figure}
\centering
\begin{tikzpicture}
\begin{axis}[axis lines=middle,ticks=none,axis line style={draw=none},view={10}{20},
colormap/viridis high res,colormap/blackwhite,mesh/interior colormap name=viridis high res,
xmin=-5,xmax=10,ymin=-5,ymax=10,zmin=-.5,zmax=1.5,trig format plots=rad,width=9cm]
\addplot3[z buffer=sort,surf,domain=5:0,y domain=0:1]({x - y*sqrt(8)}, {x + y*sqrt(8)}, y);
\addplot3[z buffer=sort,surf,domain=-0.25*pi:0.75*pi,y domain=0:1]({5 + 4*y*cos(x)}, {5 + 4*y*sin(x)}, y);
\addplot3[z buffer=sort,surf,domain=0.75*pi:1.75*pi,y domain=0:1]({4*y*cos(x)}, {4*y*sin(x)}, y);
\addplot3[z buffer=sort,surf,domain=0:5,y domain=0:1]({x + y*sqrt(8)}, {x - y*sqrt(8)}, y);
\draw[black,thick] ({sqrt(8)}, {-sqrt(8)}, 1) -- ({5 + sqrt(8)}, {5 - sqrt(8)}, 1)
             ({-sqrt(8)}, {sqrt(8)}, 1) -- ({5 - sqrt(8)}, {5 + sqrt(8)}, 1);
\draw[black,thick] (0, 0, 0) -- (5, 5, 0);
\addplot3[black,thick,domain=-0.25*pi:0.75*pi,samples y=1]({5 + 4*cos(x)}, {5 + 4*sin(x)}, 1);
\addplot3[black,thick,domain=0.75*pi:1.75*pi,samples y=1]({4*cos(x)}, {4*sin(x)}, 1);
\draw[dotted,thick] ({-sqrt(8)}, {sqrt(8)}, 1) -- ({sqrt(8)}, {-sqrt(8)}, 1)
                    ({5 - sqrt(8)}, {5 + sqrt(8)}, 1) -- ({5 + sqrt(8)}, {5 - sqrt(8)}, 1)
                    (0, 0, 1) -- (5, 5, 1);
\end{axis}
\end{tikzpicture}
\caption{The surface of a trough: at fixed \(z\), we include all points no
further than~\(4z\) from~\(e\).}
\label{fig:trough}
\end{figure}
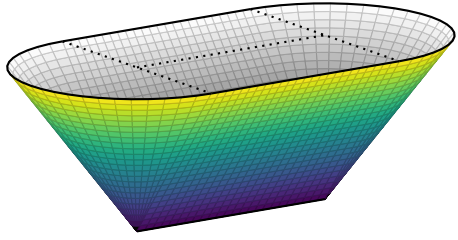

\begin{definition}\label{def:trough}
Given a segment \(e \subset \R^2\) and \(0 < \eps < 1\), define
\[\trough(e, \eps) = \{(x, y, z) \in \R^3 \mid d((x, y), e) \leq 4z \leq
\sfrac{8 \lvert e\rvert}{\eps}\}\,,\]
where \(d((x, y), e)\) is the distance from \((x, y)\) to the closest point
on \(e\).
See \cref{fig:trough}.
\end{definition}

A trough is a three-dimensional object consisting of a triangular prism and two
half-cones.
Any \(z\)-slice can be seen as a \(z\)-neighbourhood of \(e\).
We show the following \lcnamecref{lem:trough_ld}.

\begin{lemma}\label{lem:trough_ld}
Let \(P = (V, E)\) be \(\lambda\)-low density, and let \(0 < \eps < 1\).
The set \(\{\trough(e, \eps) \mid e \in E\}\) is \(k\)-low density for
\(k = \bO(\sfrac{\lambda}{\eps^2})\).
\end{lemma}
\begin{proof}
We first bound the size of \(\trough(e, \eps)\).
Let \((x, y, z) \in \trough(e, \eps)\).
Note that \(0 \leq z \leq \sfrac{2 \lvert e\rvert}{\eps}\).
Furthermore, \(d((x, y), e) \leq \sfrac{8 \lvert e\rvert}{\eps}\), so
\((x, y)\) must lie inside a disk centred at the midpoint of \(e\) with radius
\(\sfrac{9 \lvert e\rvert}{\eps}\).
Thus \((x, y, z)\) lies inside a cube with side length
\(\sfrac{18 \lvert e\rvert}{\eps}\), which bounds the size of
\(\trough(e, \eps)\).

Let \(H_r\) be an axis-parallel cube with side length \(r\), and let its
smallest \(z\)-coordinate be \(z_{\min} \geq 0\).
Suppose \(\trough(e, \eps)\) of size at least \(r\) intersects \(H_r\), and let
\((x, y, z)\) be a point in the intersection.
Let \(h\) be the projection of the centre of \(H_r\) onto the plane \(z = 0\).
Then
\[d(h, e) \leq d(h, (x, y)) + d((x, y), e) \leq r + 4z \leq 5r + 4z_{\min}\,,\]
where the first step follows by the triangle inequality, the second by
\((x, y, z)\) lying in the intersection, and the third by
\(z \leq z_{\min} + r\).
Furthermore, the size of \(\trough(e, \eps)\) is at least \(r\) and at most
\(\sfrac{18 \lvert e\rvert}{\eps}\), so
\(r \leq \sfrac{18 \lvert e\rvert}{\eps}\); and
\(z_{\min} \leq \sfrac{2 \lvert e\rvert}{\eps}\).
Therefore, \(5r + 4z_{\min} \leq \sfrac{98 \lvert e\rvert}{\eps}\).

So we know \(\lvert e\rvert \geq (5r + 4z_{\min}) \cdot \sfrac{\eps}{98}\).
By \(\lambda\)-low-density property, any ball with diameter
\((5r + 4z_{\min}) \cdot \sfrac{\eps}{98}\) is intersected by at most
\(\lambda\) such edges.
Consider a disk in \(z = 0\) with diameter \(2\cdot(5r + 4z_{\min})\) centred
at \(h\).
It can be covered by \(\sfrac{c}{\eps^2}\) smaller disks for a constant \(c\),
so there may be at most \(k = \sfrac{c\lambda}{\eps^2}\) edges close enough to
\(h\) for \(H_r\) to intersect their troughs; and so the set of troughs is
\(k\)-low density.
\end{proof}

Using the range searching data structure for low-density sets by Schwarzkopf
and Vleugels~\cite{ld_search} on the set of troughs of all edges, we obtain the
following result.

\begin{lemma}\label{lem:ld_ds}
Let \(P = (V, E)\) be a \(\lambda\)-low-density \(t\)-spanner, let
\(0 \leq \eps \leq 1\), and let \(F = \{f \in \R^2 \mid f \in e, e \in E\}\).
In \(\bO(n^2 \log n + \sfrac{\lambda}{\eps^2} \cdot n \log n)\) time and using
\(\bO(n \log^2 n + n \cdot \sfrac{\lambda}{\eps^2})\) space, we can construct
a data structure that, given a query square \(S\) in the plane with side length
\(2r\), returns a set of vertices \(T\) satisfying \cref{lem:ld_points} in time
\(\bO(\log^2 n + \sfrac{t^2}{\eps^2} + \sfrac{\lambda}{\eps^3})\).
\end{lemma}

Finally, we obtain the main result of the paper.
The proof of Gudmundsson et al.~\cite[Theorem~6.1]{prev} applies here directly,
but we use the data structures in \cref{lem:grid_all,lem:ld_ds}.

\thmmain*
\begin{proof}
See the proof by Gudmundsson et al.~\cite[Theorem~6.1]{prev}, but we use
\cref{lem:grid_all,lem:ld_ds} instead.
We analyse the space and time requirements.
For preprocessing and space, we construct the two data structures.
For distance queries, first analyse the decision version.

We query the data structure of \cref{lem:ld_ds} \(m\) times to obtain the
candidate points.
Then we construct a directed graph with
\(\bO(m \cdot (\sfrac{\lambda}{\eps^3} + \sfrac{t^2}{\eps^2})^2)\) edges.
For each edge, we do a constant number of queries to the data structure of
\cref{lem:grid_all}, each taking \(\bO(\sqrt{n}\cdot\sfrac{\lambda}{\eps})\)
time.
Finally, we decide if there is a suitable directed path in the graph.
Overall, the decision version takes \(\bO(m\sqrt{n}\cdot \sfrac{\lambda}{\eps}
\cdot(\sfrac{\lambda}{\eps^3} + \sfrac{t^2}{\eps^2})^2)\) time.

For the optimisation version, we apply parametric search using
\(N_P = m\sqrt{n}\) parallel processors.
The sequential version runs in the same time as the decision version, so
\(T_S = \bO(m\sqrt{n}\cdot \sfrac{\lambda}{\eps} \cdot
(\sfrac{\lambda}{\eps^3} + \sfrac{t^2}{\eps^2})^2)\).
In the parallel version, the steps for each of \(m\) points can be executed in
parallel; and finding the weight of an edge by querying the distance data
structure can be done with \(\sqrt{n}\) processors, each performing
\(\bO(\sfrac{\lambda}{\eps})\) amount of work, then combining the values to
find the minimum in \(\bO(\log n)\) time.
Thus, \(T_P = \bO(\log^2 n + (\sfrac{\lambda}{\eps} + \log n)\cdot
(\sfrac{\lambda}{\eps^3} + \sfrac{t^2}{\eps^2})^2)\).
The time for the optimisation version is now
\(\bO(N_P T_P + T_P T_S \log N_P)\).
Let \(\varphi = (\sfrac{\lambda}{\eps^3} + \sfrac{t^2}{\eps^2})^2\); then the
query time is
\[\bO\big(\varphi\cdot\sfrac{\lambda}{\eps} \cdot m\sqrt{n}\log mn \cdot
(\log^2 n + \varphi \log n + \varphi\cdot\sfrac{\lambda}{\eps})\big)\,.\]
Treating \(t\) and \(\lambda\) as constant, we get
\(\bO(\eps^{-7} \cdot m\sqrt{n}\log mn \cdot
(\log^2 n + \eps^{-6} \log n + \eps^{-7}))\); and treating also \(\eps\) as a
constant, the query time becomes \(\bO(m\sqrt{n}\log mn \cdot \log^2 n)\).

Finally, for the reporting query, we first run the distance query and record
the path in the graph, as well as the edges aligned with the query curve
vertices, and record the optimal order of endpoints of these edges.
Now we can simply perform the individual segment reporting queries, as before,
costing us extra \(\bO(\sfrac{\ell}{\eps})\) time for a path of
length~\(\ell\).
\end{proof}

\section{Conclusion}\label{sec:concl}
Given a \(\lambda\)-low-density \(t\)-spanner, we construct a data structure of
\(\tilde\bO(n^{\sfrac{3}{2}})\) size that answers \((1 + \eps)\)-approximate
map-matching distance queries in \(\tilde\bO(m \sqrt n)\) time, where \(n\)~is
the complexity of the graph, \(m\)~is the complexity of the query curve, and
\(\tilde\bO(\cdot)\) hides dependence on~\(\lambda\),~\(t\),~\(\eps\), and
logarithmic factors in~\(n\) and~\(m\).
Our data structure can also report the matched path in
\(\tilde\bO(m \sqrt n + \ell)\) time, where~\(\ell\) is the length of the path.

We believe that real-life world networks are \(\lambda\)-low-density and
\(t\)-spanners.
Chen et al.~\cite{refs_frechet3} verified that~\(\lambda < 30\) for road
networks in San Francisco, Athens, Berlin, and San Antonio, whereas
Nordbeck~\cite{Nordbeck64} verified that~\(t < 1.5\) for a sampled section of
Trollhätten.
Still, one may envision a scenario where \(t\)~is a larger constant, say, for
points on the opposite sides of a poorly bridged river or across national
borders.
It may be interesting to see if a different construction may require a weaker
assumption.
In general, we believe that \(\lambda\)~and~\(t\) are small constants for many
real-life road networks, but verifying this remains an open problem.

We restate several open problems posed by Gudmundsson et al.~\cite{prev} which
are also relevant to our result.
Is it possible to construct a data structure where we can select~\(\eps\) at
query time, rather than at preprocessing time?
Can we avoid using parametric search in the query procedure, to make the data
structure more practical?
Can we prove a lower bound to rule out a much more efficient data structure?

\bibliography{references}
\end{document}